\documentclass[a4paper,12pt]{article}

\usepackage{color}
\usepackage{graphicx}
\usepackage{amsmath}
\usepackage{amsfonts}
\usepackage{amsthm}
\usepackage{amscd}
\usepackage[mathcal]{euscript}
\usepackage{epsfig}
\usepackage{amssymb}
\usepackage{slashed}
\usepackage{tabularx}
\usepackage{calligra}
\usepackage{enumerate}
\usepackage{hyperref}
\usepackage{float}
\usepackage{stackrel}
\usepackage[T1]{fontenc}
\usepackage{longtable}
\usepackage{amsthm}
\usepackage{mathrsfs}
\usepackage{verbatim} 
\usepackage{fancyhdr}
\usepackage{tikz}
\usepackage{tikz-cd}
\usepackage{slashed}
\usetikzlibrary{matrix}
\usetikzlibrary{positioning}
\usepackage[all]{xy}
\usepackage{units}
\usepackage{textcomp}
\usepackage{turnstile}
\usepackage{vmargin}
\usepackage{anysize}

\numberwithin{equation}{section}

\setmargins{2,5cm}       
{1,5cm}                        
{16cm}                      
{23,42cm}                    
{10pt}                           
{1cm}                           
{0pt}                             
{2cm}                           

\newtheorem{thm}{Theorem}[section]
\newtheorem{prop}[thm]{Proposition}
\newtheorem{lem}[thm]{Lemma}

\newtheorem{remark}[thm]{Remark}

\renewcommand{\a}{\alpha}
\renewcommand{\b}{\beta}
\renewcommand{\c}{\nabla}
\newcommand{\g}{\gamma}

\newcommand{\C}{\mathscr{C}}
\newcommand{\e}{\epsilon}
\newcommand{\ve}{\varepsilon}
\newcommand{\vp}{\varphi}
\renewcommand{\l}{\lambda}
\renewcommand{\L}{\Lambda}
\renewcommand{\k}{\kappa}
\newcommand{\m}{\mu}
\newcommand{\n}{\nu}
\renewcommand{\d}{\delta}

\newcommand{\w}{\omega}
\newcommand{\Om}{\Omega}
\renewcommand{\t}{\tau}
\newcommand{\T}{\Theta}
\newcommand{\U}{\Upsilon}
\newcommand{\X}{{\sf X}}
\newcommand{\Y}{{\sf Y}}
\newcommand{\Z}{{\sf Z}}
\newcommand{\p}{\partial}
\newcommand{\M}{\mathscr{M}}
\newcommand{\mc}{\mathcal}
\newcommand{\mf}{\mathfrak}
\newcommand{\mbb}{\mathbb}
\newcommand{\R}{\mathbb{R}}
\newcommand{\s}{\mathfrak{s}}
\newcommand{\mr}{\mathring}

\newcommand{\teuk}{\ensuremath{\hspace{.05cm}\mathaccent\Box{\text{\tiny \textsc{T}}}\hspace{.05cm}}}     
\newcommand{\wt}{\widetilde}
\newcommand{\wh}{\widehat}

\DeclareMathOperator{\tho}{\text{\rm\th}}
\DeclareMathOperator{\edt}{\text{\rm\dh}}

\begin{document}

\title{Conformal invariance, complex structures, and the Teukolsky connection}

\author{Bernardo Araneda\footnote{E-mail: \texttt{baraneda@famaf.unc.edu.ar}} \\ 
 \\
Facultad de Matem\'atica, Astronom\'{\i}a y F\'{\i}sica\\
Universidad Nacional de C\'ordoba\\ 
Instituto de F\'{\i}sica Enrique Gaviola, CONICET\\
Ciudad Universitaria, (5000) C\'ordoba, Argentina 
}

\date{May 29, 2018}

\maketitle

\begin{abstract}
We show that the Teukolsky connection, which defines generalized wave operators governing the behavior 
of massless fields on Einstein spacetimes of Petrov type D, has its origin in a distinguished conformally and 
GHP covariant connection on the conformal structure of the spacetime. 
The conformal class has a (metric compatible) integrable almost-complex structure under which the Einstein 
space becomes a complex (Hermitian) manifold. 
There is a {\em unique} compatible Weyl connection for the conformal structure,
and it leads to the construction of a conformally covariant GHP formalism and a generalization of it 
to weighted spinor/tensor fiber bundles. In particular,
`weighted Killing spinors', previously defined with respect to the Teukolsky connection, are shown to have 
their origin in the GHP-Weyl connection, and we show that the type D
principal spinors are actually {\em parallel} with respect to it.
Furthermore, we show that the existence of a conformal Killing-Yano tensor can be thought to be a 
consequence of the presence of a K\"ahler metric in the conformal class.
These results provide an interpretation of the persistent hidden symmetries appearing in black hole 
perturbations.
We also show that the preferred Weyl connection allows a natural injection of spinor fields into local 
twistor space and that this leads to the notion of {\em weighted local twistors}.
Finally, we find conformally covariant operator identities for massless fields and the corresponding 
wave equations.
\end{abstract}

\section{Introduction}

Let $(\M,g_{ab})$ be a 4-dimensional spacetime with a spinor structure and Levi-Civita connection $\c_a$.
Zero rest mass free fields of spin $\s\in\mbb{N}_0/2$ are totally symmetric spinors 
$\vp_{A_1...A_{2\s}}=\vp_{(A_1...A_{2\s})}$ satisfying the field equations
\begin{align}
 \c^{A_1'A_1}\vp_{A_1...A_{2\s}}=0, \qquad \text{for } \s>0, \label{feq} \\
 \Box\vp=0, \qquad \text{for } \s=0, \label{weq}
\end{align}
where $\Box=g^{ab}\c_{a}\c_{b}$.
In Minkowski spacetime $\mbb{M}=\R^{1,3}$, these equations have their origin in the study of the massless 
irreducible representations of the universal covering of the Poincar\'e group. 
In an arbitrary curved spacetime the most interesting cases of the equations
correspond to spin $0$, $1/2$, $1$ and $2$, 
and our interest on them comes from studies related to the black hole stability problem.
Penrose has shown in \cite{Penrose1965} that all solutions of (\ref{feq}) in $\mbb{M}$
can be obtained from solutions of the wave equation (\ref{weq}). The process makes use of covariantly
constant spinor fields, so it does not generalize to arbitrary curved spacetimes. 
An alternative mechanism of `spin lowering/raising', potentially generalizable to other spacetimes,
is obtained from the use of {\em Killing spinors} \cite{Penrose2}, 
which are totally symmetric spinor fields $\sigma^{A_1...A_n}$ satisfying
\begin{equation}\label{kseq}
 \c_{A'}{}^{(A}\sigma^{B_1...B_n)}=0.
\end{equation}
If $\sigma^{A_1...A_n}$ is a solution to this equation, and $\vp_{A_1...A_n}$ solves (\ref{feq}), 
then (in Minkowski) $\Phi=\sigma^{A_1...A_n}\vp_{A_1...A_n}$ solves $\Box\Phi=0$.
These objects are closely related to twistor theory; in particular, the `primary spinor part' of 
any symmetric twistor ${\sf Z}^{\a_1...\a_n}$ (see section \ref{sec-wlt} for notation)
in $\mbb{M}$ satisfies (\ref{kseq}). 
Killing spinors are also closely associated to the notion of `hidden symmetries' in General Relativity, since 
the tensorial version of (\ref{kseq}) is related to the conformal Killing-Yano equation.
More precisely, Einstein spacetimes of Petrov type D have a 2-index solution $K^{AB}$ to (\ref{kseq}),
and this object (whose tensorial analogues are conformal Killing-Yano tensors) is responsible for the 
integrability of geodesic motion and the separability of the Klein-Gordon and Dirac equations on the Kerr 
spacetime.

In a curved spacetime, on physical grounds one would expect that, for massless fields of spin 1 and 2, it might 
be possible to encode the whole dynamics of the field in a complex scalar field, since it is well known that the 
dynamical parts of the electromagnetic and gravitational fields have only two degrees of freedom. 
Instead of the usual wave equation (\ref{weq}), the complex scalar field would satisfy a `more sophisticated' 
wave-like equation deduced from (\ref{feq}), in which $\Box$ is replaced by some normally hyperbolic 
operator\footnote{i.e. a second order linear differential operator whose principal symbol is
the spacetime metric.} plus some potential related to the curvature. 
In the context of black hole perturbations, the background spacetime belongs to the type D class, 
so that there are two preferred null directions $o^A,\iota^A$ singled out by the geometry, and a unique non-trivial 
Weyl scalar $\Psi_2$. The GHP formalism \cite{GHP} is then especially suited for this situation.
The wave-like equations are known as Teukolsky, Regge-Wheeler, Fackerell-Ipser, etc.; the specific 
equation depends on the spin $\s$ and spin weight $s$ of the variables considered, and has the general form
\begin{equation}\label{teukeq}
 (\teuk_{2s}+V_{\s,s})[\Psi^k_2 \; \Phi^{(s)}]=0,
\end{equation}
where $k\in\mbb{Z}/3$, $V_{\s,s}$ is a complex potential, and $\Phi^{(s)}$ is the spin weight $s$ component of 
$\vp_{A_1...A_{2\s}}$ in the principal dyad.
The weighted wave operator $\teuk_p$ acting on GHP type $\{p,0\}$ quantities is defined as 
\begin{equation}
 \teuk_p:=g^{ab}D_{a}D_{b},
\end{equation}
where $D_{a}:=\T_{a}+pB_{a}$ (for $p\in\mbb{Z}$) is a modified GHP covariant derivative, with $\T_a$ the 
usual GHP derivative (see e.g. \cite{GHP}) and $B_a:=-\rho n_a+\t\bar{m}_a$ a modified GHP connection 1-form introduced 
in \cite{Andersson2011}. From now on we will refer to $D_a$ as the {\em Teukolsky connection}.
The form (\ref{teukeq}) of the equations makes manifest their wave-like nature, which ensures the well-known 
properties of solutions of this kind of equations such as existence, uniqueness, and Cauchy stability in 
globally hyperbolic regions.
Furthermore, similarly to the flat case, the field $\Psi^k_2\Phi^{(s)}$ solving (\ref{teukeq}) is of the form 
$\Psi^k_2\Phi^{(s)}=P^{A_1...A_{2\s}}_s\vp_{A_1...A_{2\s}}$, where now $P^{A_1...A_{2\s}}_s$ is given by
\begin{equation}\label{P}
 P^{A_1...A_{2\s}}_s=\Psi^{-(\s-s)/3}_2\iota^{(A_1}...\iota^{A_{\s-s}}o^{A_{\s-s+1}}...o^{A_{2\s})}.
\end{equation}
Interestingly enough, it was found in \cite{Araneda2017} that (\ref{P}) solves the equation
\begin{equation}\label{mks}
 D_{A'}{}^{(A}P^{B_1...B_{2\s})}_s=0,
\end{equation}
and therefore it can be considered as a Killing spinor with respect to the Teukolsky connection 
(in \cite{Araneda2017} these objects were called `weighted' Killing spinors). 
The ordinary Killing spinor $K^{AB}=\Psi^{-1/3}_2o^{(A}\iota^{B)}$ of type D spacetimes mentioned above 
is just a particular case of (\ref{P})-(\ref{mks}).
This fact suggests that there should be a deeper geometrical understanding of the Teukolsky connection, 
and of its relation with the `generalized' hidden symmetries coming from (\ref{mks}), that (as far as we know)
seems to be currently unknown. 
Since the Teukolsky equations (\ref{teukeq}) are central to the linear approach to the black hole 
stability problem (one of the major open issues in classical General Relativity),
one of the main goals of this article is to provide such a geometrical understanding.

Now, the relation between (\ref{teukeq}) and (\ref{feq}) in the context of perturbations 
of a type D Einstein spacetime can be formulated as follows.
Consider a family of 4-dimensional Lorentzian manifolds $(\M,g_{ab}(\ve))$ such that $g_{ab}(0)$ is an
 Einstein spacetime 
of Petrov type D. In order to get a second order, scalar equation from (\ref{feq}), one has to apply a first order 
differential operator contracting all spinor indices. 
Let $\s=\frac{1}{2}, 1, 2$ and $s=0,\pm \s$. It was found in \cite{Araneda2016} that, without assuming that 
any field equations are satisfied, one has the operator identity
\begin{equation}\label{mainid0}
 P^{A_1...A_{2\s}}_{s}(\c_{A'_1A_1}-2\s A_{A'_1A_1})\c^{A'_1B}\vp_{A_2...A_{2\s} B}\doteq
 (\teuk_{2s}+V_{\s,s})[P^{A_1...A_{2\s}}_s\vp_{A_1...A_{2\s}}],
\end{equation}
where ``$\doteq$'' means equality up to linear order in $\ve$ (of course this is only needed for spin $\s=2$, in 
which the corresponding field represents the linearized Weyl curvature spinor; see \cite{Araneda2016} for details), 
and the 1-form $A_a$ on the left is given 
by\footnote{In \cite{Araneda2016} this 1-form was denoted as $-A_a$, and in \cite{Araneda2017} as $-\g_a$.}
\begin{equation}\label{Apsi2}
 A_{a}=\Psi^{-1/3}_2\c_{a}\Psi^{1/3}_2.
\end{equation}
As we want to understand the geometrical structure of (\ref{teukeq}) and its relation to (\ref{mks}), 
the identity (\ref{mainid0}) turns out to be appropriate for analyzing this problem since it holds 
independently of the field equations.
From this point of view we must also understand the left hand side of (\ref{mainid0}) in geometrical terms, 
which in particular raises the question of what is the interpretation of the 1-form (\ref{Apsi2});
namely, is it a connection form on some bundle, and if so, related to which symmetry? 
Moreover, a complete understanding of the symmetry structures underlying (\ref{mainid0}) 
should give precise relations (geometrical or 
otherwise) between all objects involved in it: the spinor $P^{A_1...A_{2\s}}_s$ solving (\ref{mks}), 
the 1-form $A_a$, and the Teukolsky connection $D_a$ and its associated wave operator $\teuk_p$.
The aim of this work is to answer all these questions.

\subsection{Main results and overview}

In this article we prove that all the just mentioned questions can be answered by taking into account 
a simple observation: the field equations (\ref{feq}) are conformally invariant. 

An appropriate setting for studying the present problem is then the conformal manifold $(\M,[g_{ab}])$ 
associated to our original spacetime. The analogue of the Levi-Civita connection is a Weyl connection 
for the conformal structure. 
Type D spacetimes have a naturally defined, metric compatible,
almost-complex structure $J$ (which is actually integrable, thus these are complex Hermitian
manifolds), and there is a {\em unique} Weyl connection compatible with it. 
In section \ref{sec-cs} we show that (\ref{Apsi2}) arises naturally in this way.

The (integrable) almost-complex structure is actually parallel with respect to the Levi-Civita connection of 
a particular member of the conformal class, thus this member is a K\"ahler manifold. In section 
\ref{sec-kahler} we show that this implies the existence of a conformal Killing-Yano tensor 
(equivalently, a Killing spinor) in the original spacetime.

Furthermore, fixing a pair of null directions on the conformal manifold and using the preferred Weyl connection, 
in section \ref{sec-ccghp} we construct a conformally covariant GHP formalism generalized to be valid for all 
spinor/tensor conformal densities, and we show that the naturally induced covariant derivative on associated 
vector bundles {\em is precisely} the Teukolsky connection. 

The natural Weyl-GHP connection also allows us to introduce a geometrically well defined and meaningful 
notion of {\em weighted Killing spinors}, this gives (\ref{mks}) as a consequence of the stronger result 
that the canonical type D principal spinors are {\em parallel} in this more general setting (no powers of $\Psi_2$ 
needed); we prove this in section \ref{sec-wks}.

In section \ref{sec-wlt} the preferred Weyl connection is shown to give a natural injection of spinor 
fields into local twistor space (and a corresponding conformally invariant short exact sequence).
This way we are immediately led to the consideration of {\em weighted local twistors} and to a 
re-derivation of the concept of weighted Killing spinors.

We then find in section \ref{sec-cci} conformally covariant operator identities on spinor fields, which 
give the relation between field equations and conformally-GHP covariant wave equations, thus giving 
a version of (\ref{mainid0}) in a `unified' geometrical framework.

We show the explicit relation with Teukolsky wave operators in section \ref{sec-rt}, and with 
a generalized Laplace-de Rham operator (introduced in \cite{Araneda2017}) in section \ref{sec-laplace}.

Appendix \ref{appendix} contains some details of the calculations left out of the main text 
in section \ref{sec-spin2}.

\subsection{Notation and conventions}

We will work on (conformal) 4-dimensional Lorentzian spin spacetimes, with metric signature $(+---)$. 
Our conventions for tensor, spinor and twistor indices (and for curvature tensors, etc.)
are the same as those of Penrose $\&$ Rindler \cite{Penrose1, Penrose2}. We will denote as 
$\mbb{S}^{A... B'...}_{C... D'...}$ the space of spinor fields (i.e. sections of spinor bundles) of 
the type indicated by the indices.
We will furthermore use the notation $\mbb{S}^{A... B'...}_{\{p,q\}\;C... D'...}[w]$ for spinor fields with 
conformal weight $w$ and GHP type $\{p,q\}$.

\section{Conformal and complex structures}

\subsection{Preliminaries}

We recall some basic facts about conformal geometry, since this will be useful in the following sections.
Let $(\M,g_{ab})$ be a 4-dimensional spin spacetime with spinor metric $\e_{AB}$. 
The conformal manifold associated to this spacetime is $(\M,[g_{ab}])$, where $[g_{ab}]$ is the 
equivalence class of 
metrics that are conformally related to $g_{ab}$, i.e. $\wh{g}_{ab}\in[g_{ab}]$ if and only if there exists 
a smooth, nowhere vanishing positive scalar function $\Om$ such that $\wh{g}_{ab}=\Om^2 g_{ab}$. 
The map $g_{ab}\mapsto\wh{g}_{ab}=\Om^2 g_{ab}$ is called conformal transformation. 
For the spinor metric, this is equivalent to $\e_{AB}\mapsto\wh\e_{AB}=\Om \e_{AB}$. 
For the inverse spacetime and spinor metrics, we have $\wh{g}^{ab}=\Om^{-2}g^{ab}$ and 
$\wh\e^{AB}=\Om^{-1}\e^{AB}$.
If $\{o^A,\iota^A\}$ is a spin frame (i.e. $\e_{AB}=2o_{[A}\iota_{B]}$), the spinors transform as
\begin{equation}\label{ctsf}
 \wh{o}^A=\Om^{w_0}o^A, \qquad \wh{\iota}^A=\Om^{w_1}\iota^A,
\end{equation}
for some numbers $w_0,w_1\in\R$ such that $w_0+w_1+1=0$. 
For the associated null tetrad $\{\ell^a,n^a,m^a,\bar{m}^a\}$ this is equivalent to
\begin{equation}
 \wh{\ell}^a=\Om^{2w_0}\ell^a, \qquad \wh{n}^a=\Om^{2w_1}n^a, \qquad \wh{m}^a=\Om^{w_0+w_1}m^a, 
 \qquad \wh{\bar{m}}{}^a=\Om^{w_0+w_1}\bar{m}^a.
\end{equation}
We will eventually choose particular values for $w_0$ and $w_1$.

Metrics conformally related to $g_{ab}$ can be viewed as a subbundle $\mc{Q}\subset T^{*}\M\odot T^{*}\M$ 
with fibers $\R^{+}$, which in turn can be understood as a principal bundle over $\M$ with structure 
group\footnote{we denote as $\R^{\times}$ (resp. $\R^{+}$) the multiplicative group of real (resp. positive) 
numbers.} $\R^{+}$ (see e.g. \cite[Section 2.4]{Curry2014}). 
Then one can construct vector bundles associated to $\mc{Q}$ known as conformally weighted line bundles, 
denoted by $\mc{E}[w]$, where $w\in\R$ is called conformal weight, and whose elements transform under 
conformal rescaling as $\wh\phi=\Om^w\phi$. Sections of $\mc{E}[w]$ are known as conformal densities 
of weight $w$.
More generally, if $\mbb{E}$ is a vector bundle over $\M$, one can construct the weighted bundle 
$\mbb{E}\otimes\mc{E}[w]=:\mbb{E}[w]$, the sections of which will be conformally weighted spinor/tensor fields.
In particular, $\mbb{S}^{A...B'...}_{C...D'...}[w]$ will denote the space of spinor fields with conformal weight $w$.

If $\c_{a}$ is the Levi-Civita connection of $g_{ab}$, then under a conformal transformation, the Levi-Civita 
connection of the new metric $\wh{g}_{ab}=\Om^2 g_{ab}$ acting on a tensor $T^{b_1...b_k}_{c_1...c_l}$ is
\begin{align}
\nonumber \wh{\c}_a T^{b_1...b_k}_{c_1...c_l} = \;\c_aT^{b_1...b_k}_{c_1...c_l} & 
 +K_{a}{}^{b_1}{}_{d}T^{d...b_k}_{c_1...c_l}+...+K_{a}{}^{b_l}{}_{d}T^{b_1...d}_{c_1...c_l} \\
 & -K_{a}{}^{d}{}_{c_1}T^{b_1...b_k}_{d...c_l}-...-K_{a}{}^{d}{}_{c_l}T^{b_1...b_k}_{c_1...d}, \label{ctd}
\end{align}
where
\begin{equation}
 K_{a}{}^{b}{}_{c}:=g^{bd}(\U_{a}g_{dc}+\U_{c}g_{da}-\U_{d}g_{ab}),
\end{equation}
with $\U_{a}:=\Om^{-1}\c_{a}\Om$. 
For a spinor $\Psi^{B_1...B_k}_{C_1...C_l}$, the corresponding formula is
\begin{align}
\nonumber \wh{\c}_a \Psi^{B_1...B_k}_{C_1...C_l}= \c_a\Psi^{B_1...B_k}_{C_1...C_l} & 
 + \L_{aD}{}^{B_1} \Psi^{D...B_k}_{C_1...C_l}+...+ \L_{aD}{}^{B_k} \Psi^{B_1...D}_{C_1...C_l} \\
 & -\L_{aC_1}{}^{D}\Psi^{B_1...B_k}_{D...C_l}-...-\L_{aC_l}{}^{D}\Psi^{B_1...B_k}_{C_1...D}, \label{ctd-spinor}
\end{align}
where
\begin{equation}
 \L_{aC}{}^{B}:=\U_{A'C}\e_{A}{}^{B}.
\end{equation}
For spinors with primed indices, the corresponding formula is deduced from (\ref{ctd-spinor}) 
by complex conjugation, and taking into account that the 1-form $\U_a$ is real.
The relation between $K_{a}{}^{b}{}_{c}$ and $\L_{aC}{}^{B}$ is given by
$K_{a}{}^{b}{}_{c}=\L_{a C}{}^{B}\bar\e_{C'}{}^{B'}+\bar{\L}_{a C'}{}^{B'}\e_{C}{}^{B}$.
See \cite[section 4.4]{Penrose1}.

The various parts of the curvature of $g_{ab}$ have different behaviors under conformal transformations.
In particular, the Weyl spinor can be shown to be conformally invariant, 
\begin{equation}
 \wh{\Psi}_{ABCD}=\Psi_{ABCD}
\end{equation}
(see e.g. \cite[section 6.8]{Penrose2}), 
and therefore its algebraic structure is common to all metrics in the conformal class $[g_{ab}]$. 
This can also be seen from the conformal behavior of the GHP spin coefficients: 
$\k,\k',\sigma,\sigma'$ are conformal densities, and by the Goldberg-Sachs theorem (valid for vacuum 
solutions with cosmological constant), algebraic speciality is equivalent to the existence of a 
shear-free null geodesic congruence:
\begin{equation}\label{as}
 \k=\sigma=0.
\end{equation}
(In particular, type D spaces have $\k=\k'=\sigma=\sigma'=0$.)
The conformal behavior of the other parts of the curvature can be conveniently described by means of the 
{\em Schouten tensor}, which according to our conventions (\cite{Penrose2}) is defined by
\begin{equation}\label{schouten}
 P_{ab}:=-\tfrac{1}{2}(R_{ab}-\tfrac{R}{6}g_{ab}),
\end{equation}
and whose relation with the Riemann tensor is $R_{ab}{}^{cd}=C_{ab}{}^{cd}+4P_{[a}{}^{[c}g_{b]}{}^{d]}$.
Under conformal transformations, $P_{ab}$ changes as
\begin{equation}\label{CTschouten}
 \wh{P}_{AA'BB'}=P_{AA'BB'}-\c_{BB'}\U_{AA'}+\U_{AB'}\U_{BA'}.
\end{equation}

\subsubsection{Weyl connections}

Below we will need the concept of a {\em Weyl connection} for the conformal structure associated to $g_{ab}$,
which is a torsion-free connection $\slashed{\c}_{a}$ such that
\begin{equation}\label{wconn}
 \slashed{\c}_{a}g_{bc}=-2f_{a}g_{bc}
\end{equation}
for some 1-form $f_{a}$. Assuming $\slashed{\c}_a$ to be fixed and calculating the conformal transformation 
of (\ref{wconn}), we find that $f_a$ has the conformal behavior
\begin{equation}
 f_{a} \mapsto \wh{f}_a= f_{a}-\U_{a}.
\end{equation}
We allow $f_{a}$ to be a complex 1-form (see section \ref{sec-cs}). 
When acting on tensors, the relation between $\slashed{\c}_a$ and the Levi-Civita connection is given by a 
formula analogous to (\ref{ctd}), but replacing $K_{a}{}^{b}{}_{c}$ with the tensor
\begin{equation}\label{Qcc}
 Q_{a}{}^{b}{}_{c}:=g^{bd}(f_{a}g_{dc}+f_{c}g_{da}-f_{d}g_{ab}).
\end{equation}
The action of $\slashed{\c}_a$ on spinors is analogous to (\ref{ctd-spinor}) and its complex conjugate, but 
replacing $\L_{aC}{}^{B}$ and $\bar{\L}_{aC}{}^{B}$ respectively for
\begin{equation}\label{Wcc}
 W_{a C}{}^{B}=f_{A'C}\e_{A}{}^{B}, \qquad  \wt{W}_{a C'}{}^{B'}=f_{AC'}\bar\e_{A'}{}^{B'}.
\end{equation}
We must use the object $\wt{W}_{a C'}{}^{B'}$ instead of $\bar{W}_{a C'}{}^{B'}$ because the 1-form $f_a$ 
can be complex, i.e. $\bar{f}_a\neq f_{a}$
(note that in the case $f_a$ real, we have $\wt{W}_{a C'}{}^{B'}=\bar{W}_{a C'}{}^{B'}$).

If $\Psi$ is an arbitrary spinor/tensor field with conformal weight $w\neq0$ and $\slashed{\c}_a$ is a 
Weyl connection, then $\slashed{\c}_{a}\Psi$ is not a conformal density.
A covariant derivative operator that maps conformal densities of weight $w$, to conformal densities 
(of the same weight), can be constructed as
\begin{equation}\label{ccd}
 \C_{a}\Psi=\slashed{\c}_{a}\Psi+wf_{a}\Psi.
\end{equation}
It then follows that $\wh{\C_{a}\Psi}=\Om^{w}\C_{a}\Psi$. We say that $\C_{a}$ is a {\em conformally 
covariant derivative}. More generally, we will say that a linear differential operator $P$ is {\em conformally 
covariant} if, when acting on conformal spinor/tensor densities $\Psi$ of weight $w$, it satisfies 
$\wh{P(\Psi)}=\Om^{w'}P(\Psi)$ for some $w'\in\R$.

\subsection{Almost-complex structure}\label{sec-cs}

Consider the conformal manifold $(\M,[g_{ab}])$ associated to $g_{ab}$ (whose Levi-Civita connection is $\c_a$). 
Suppose that there exists an almost-complex structure $J$: a linear map $J:T\M\to T\M$ such that 
$J^2=-\rm{Id}$ (the identity map on $T\M$). 
Assume also that $J$ is metric compatible, i.e. $g(JX,JY)=g(X,Y)$ for all $X,Y\in T\M$. 
In index notation, $J_{a}{}^{b}$ is a compatible almost-complex structure if
$J_{a}{}^{c}J_{c}{}^{b}=-\d^{b}_{a}$ and $J_{a}{}^{c}J_{b}{}^{d}g_{cd}=g_{ab}$.
(Note that these  conditions imply $J_{(ab)}=0$.) 
The resulting triple $(\M,[g_{ab}],J)$ is called a {\em conformal almost-Hermitian manifold}.
As discussed in \cite{Bailey1991} (see also e.g. \cite{Gover2013}), 
there is a {\em unique} Weyl connection $\slashed{\c}_a$ on $(\M,[g_{ab}],J)$ that is compatible with $J$, 
i.e. such that 
\begin{equation}
 \slashed{\c}_{b}J_{a}{}^{b}=0.
\end{equation}
Acting on an arbitrary tensor field, such a connection is given by a formula analogous to (\ref{ctd}) 
but replacing $K_{a}{}^{b}{}_{c}$ with (\ref{Qcc}), where the 1-form $f_a$ is canonically found to be
\begin{equation}\label{uwc}
 f_a=-\tfrac{1}{2}J_{b}{}^{c}\c_{c}J_{a}{}^{b}.
\end{equation}
Now, for any choice of $g_{ab}\in[g_{ab}]$, a spin frame $\{o^A,\iota^A\}$, $o_A\iota^A=1$, 
or equivalently the associated null tetrad $\{N^a_{\bf a}\}=\{\ell^a,n^a,m^a,\bar{m}^a\}$ (${\bf a}=0,...,3$), 
always defines an almost-complex structure as\footnote{note that (\ref{acs}) is a {\em complex-valued} tensor. 
For this reason, in \cite{Flaherty1976} it is referred to as a `modified' almost-complex structure.
A {\em real} almost-complex structure can be constructed as 
$J_{a}{}^{b}=-\ell_a\ell^b+n_a n^b-im_a\bar{m}^b+i\bar{m}_a m^b$, but this has a number of undesirable 
properties that make it unsuitable for our purposes, see \cite[Chapter VIII]{Flaherty1976}.}
\begin{equation}\label{acs}
 J_{a}{}^{b}=i(\ell_a n^b-n_a\ell^b+\bar{m}_a m^b-m_a\bar{m}^b).
\end{equation}
As far as we know, this structure was first found in \cite{Flaherty1974}; then rediscovered in 
\cite{Bailey1991} (see also \cite{Flaherty1976} and \cite[II.2.10]{Mason1995}).
In spinor form, (\ref{acs}) is $J_{a}{}^{b}=i(o_A\iota^B+\iota_A o^B)\bar{\e}_{A'}{}^{B'}$. 
Although (\ref{acs}) is well defined in any (4-dimensional) spacetime, 
it is particularly significant for the case of Petrov type D, since then we have a pair of null directions 
$o^A,\iota^A$ determined by the geometry and $J_{a}{}^{b}$ is naturally associated to this structure.

\begin{remark}[{\bf Chirality}]\label{remarkCh}
Note that we are choosing the spinors $o^A,\iota^A$ instead of $\bar{o}^{A'}, \bar\iota^{A'}$; 
the latter choice would give the equally valid almost-complex structure 
$\tilde{J}_{a}{}^{b}=i(\bar{o}_{A'}\bar\iota^{B'}+\bar\iota_{A'}o^{B'})\e_{A}{}^{B}$, and the corresponding 
1-form (\ref{uwc}) would now be $\bar{f}_a$. In this sense 
we can consider our treatment as {\em chiral}, since left- and right-handed spinors are not treated 
on an equal footing. This is because we want to study the {\em left-handed} fields (\ref{feq}).
\end{remark}

Note that (\ref{acs}) does not depend on the representatives of 
the conformal and GHP classes (i.e. it is invariant under conformal and GHP transformations). 
Note also that $J_{a}{}^{b}$ is compatible with the conformal structure: for any $g_{ab}\in[g_{ab}]$, we have 
$J_{a}{}^{c}J_{b}{}^{d}g_{cd}=g_{ab}$, therefore a type D spacetime is a conformal almost-Hermitian manifold 
with almost-complex structure (\ref{acs}).
The associated (unique) compatible Weyl connection (\ref{uwc}) is calculated to be
\begin{equation}\label{A}
 f_{a} =\rho n_{a}+\rho' \ell_{a}-\t\bar{m}_{a}-\t' m_{a}.
\end{equation}
Under conformal rescalings, we have of course $f_a \mapsto \wh{f}_a=f_a-\U_a$.

\begin{remark}[{\bf The 1-form (\ref{Apsi2})}]\label{remarkA}
 In an Einstein spacetime of Petrov type D, the Bianchi identities in GHP form are
 $\tho\Psi_2=3\rho\Psi_2$ and $\edt\Psi_2=3\t\Psi_2$ (together with the primed versions), 
 and they imply that the 1-form (\ref{A}) {\em coincides exactly} with (\ref{Apsi2}), 
 i.e. $f_a \equiv \Psi^{-1/3}_2\c_{a}\Psi^{1/3}_2 \equiv A_a$. 
 It is important to emphasize, however, that the Bianchi identity $\c^{A'A}\Psi_{ABCD}=0$ 
 {\em is not} conformally invariant; 
 therefore by writing (\ref{A}) in this form one is explicitly breaking conformal covariance.
\end{remark}

\subsection{K\"ahler structure}\label{sec-kahler}

An almost-complex manifold is actually a complex manifold if the almost-complex structure is integrable.
By the Newlander-Nirenberg theorem, the almost-complex structure is integrable if and only 
if the associated Nijenhuis tensor ${}^{J}N(X,Y)$ vanishes for all $X,Y\in T\M$, where
${}^{J}N(X,Y)\equiv [X,Y]+J[JX,Y]+J[X,JY]-[JX,JY]$ and $[\cdot,\cdot]$ is the Lie bracket of vector fields. 
In index notation, this is ${}^{J}N_{bc}{}^{a}=-J_{b}{}^{d}\p_{d}J_{c}{}^{a}+J_{c}{}^{d}\p_{d}J_{b}{}^{a}
 + J_{d}{}^{a}(\p_{b}J_{c}{}^{d}-\p_{c}J_{b}{}^{d})$. Note that, in this expression, 
 $\p_a$ can be replaced by any torsion-free 
connection; in particular, using the Weyl connection $\slashed{\c}_a$ applied to (\ref{acs}) we find
\begin{equation}\label{scj}
 \slashed{\c}_{a} J_{b}{}^{c}=2i\bar\e_{B'}{}^{C'}[(-\k n_a+\sigma\bar{m}_a)\iota_B\iota^C 
 +(-\k' \ell_a+\sigma' m_a)o_B o^C].
\end{equation}
The Nijenhuis tensor then vanishes for a type D spacetime (see (\ref{as})), 
thus the almost-complex structure is integrable and {\em the space can be regarded as a complex manifold}. 
Complex coordinates are obtained by integrating linear combinations of the type $(1,0)$ and $(0,1)$ forms
with respect to the decomposition\footnote{this decomposition is valid because the complex tensor $J_{a}{}^{b}$ 
has eigenvalues $+i,+i,-i,-i$; for this reason one does not consider $i\d_{a}{}^{b}$ as an almost-complex 
structure, since it has eigenvalues $+i,+i,+i,+i$ (and because it is not metric compatible).} 
$(T^{*}\M)^{\mbb{C}}=T^{*}\M^{+}\oplus T^{*}\M^{-}$ induced by $J$: 
$T^{*}\M^{+}$ is spanned by $\{\ell_a dx^a, m_a dx^a\}$, and $T^{*}\M^{-}$ by $\{n_a dx^a, \bar{m}_a dx^a\}$.
Now, since $J$ is compatible with $g_{ab}$, we actually get a Hermitian manifold. 
The 2-form defined by $\w(X,Y):=g(JX,Y)$, namely $\w_{ab}=J_{ab}$, is called {\em K\"ahler form}. 
A {\em K\"ahler manifold} is a Hermitian manifold $(\M,g_{ab})$ for which the K\"ahler 
form is closed, and it is a theorem that $\w$ is closed if and only if the almost-complex structure $J$ is 
parallel with respect to the Levi-Civita connection of $g_{ab}$. It can be checked that in general $J$ is not parallel 
with respect to $\c_a$; but from (\ref{scj}) we know that, for type D spacetimes, $J$ is parallel with 
respect to $\slashed{\c}_a$. Now, it is well known that a Weyl connection is actually the Levi-Civita connection
of some $\slashed{g}{}_{ab}\in[g_{ab}]$ if the 1-form $f_a$ is closed. For Einstein type D spaces, 
by Remark \ref{remarkA}
we know that $f_a=\c_a\log\Psi^{1/3}_2$ and hence it is closed, thus $\slashed{\c}_a$ is the Levi-Civita
connection of $\slashed{g}{}_{ab}:=\Psi^{2/3}_2 g_{ab}$. Therefore, $(\M,\slashed{g}{}_{ab})$ is a 
K\"ahler manifold, i.e. {\em Einstein spacetimes of Petrov type D are conformal (with generally complex 
conformal factor) to K\"ahler spaces} (see also \cite{Flaherty1976}).

We will now show that the existence of a K\"ahler metric in the conformal structure 
is directly related to the hidden symmetries associated to conformal Killing-Yano tensors and 
Killing spinors. A conformal Killing-Yano tensor is a 2-form $Z_{ab}$ satisfying the equation 
\begin{equation}\label{cky}
 \c_{(a}Z_{b)c}=g_{ab}\xi_c-g_{c(a}\xi_{b)},
\end{equation}
where (in dimension $d$) $\xi_a=-\tfrac{1}{(d-1)}\c_bZ_{a}{}^{b}$.
We have:

\begin{lem}\label{lem-kahler}
Let $(\M,[g_{ab}],J)$ be a $d$-dimensional conformal Hermitian manifold, and assume that there exists 
a K\"ahler metric $\wt{g}_{ab}$ in the conformal class. Then the metric $g_{ab}=\Om^{-2}\wt{g}_{ab}$
admits a conformal Killing-Yano tensor, given by $Z_{ab}=\Om^{-1}J_{ab}$, where $J_{ab}:=g_{bc}J_{a}{}^{c}$.
\end{lem}

\begin{proof}
Let $\wt{\c}_a$ be the Levi-Civita connection of the K\"ahler metric $\wt{g}_{ab}$, and let $\c_a$ be 
the corresponding one to the conformally related metric $g_{ab}=\Om^{-2}\wt{g}_{ab}$. 
In what follows, raising and lowering of indices is performed with $g_{ab}$ and its inverse $g^{ab}$.
By hypothesis we have $\wt{\c}_aJ_{b}{}^{c}=0$; therefore symmetrizing and using the relation (\ref{ctd}), 
we get
\begin{equation}
 \wt{\c}_{(a}J_{b)}{}^{c}=0=\Om\c_{(a}[\Om^{-1}J_{b)}{}^{c}]+g_{ab}\U^d J_{d}{}^{c}+g_{(a}{}^{c}J_{b)}{}^{d}\U_d.
\end{equation}
Using that $\Om\c_a[\Om^{-1}J_{b}{}^{a}]=-(d-1)J_{b}{}^{a}\U_a$ and defining $Z_{a}{}^{b}=\Om^{-1}J_{a}{}^{b}$, 
it is straightforward to show that the previous equation is equivalent to
\begin{equation}
 \c_{(a}Z_{b)}{}^{c}=-\tfrac{1}{(d-1)}g_{ab}\c_{d}Z^{cd}+\tfrac{1}{(d-1)}g_{(a}{}^{c}\c_{|d|}Z_{b)}{}^{d},
\end{equation}
thus $Z_{ab}$ is a conformal Killing-Yano tensor.
\end{proof}

In 4 dimensions, the 2-form $Z_{ab}$ in Lemma \ref{lem-kahler} can be written in spinor language 
as $Z_{ab}=\psi_{AB}\bar{\e}_{A'B'}+\chi_{A'B'}\e_{AB}$ for some symmetric spinors $\psi_{AB}$ and 
$\chi_{A'B'}$, and the conformal Killing-Yano equation (\ref{cky}) implies that $\psi_{AB}$ and 
$\chi_{A'B'}$ are Killing spinors.
Thus, the existence of these objects can be thought to be a consequence of the presence 
of a K\"ahler metric in the conformal class of the spacetime.

\subsection{Conformally covariant GHP formalism}\label{sec-ccghp}

Consider a fixed spacetime $(\M,g_{ab})$.
The GHP formalism is a framework suited for spacetimes in which two null directions can be tied to the geometry.
One adapts a null frame $\{N^a_{\bf a}\}=\{\ell^a,n^a,m^a,\bar{m}^a\}$ 
to the structure of null directions, in such a way that the remaining freedom for selecting 
a frame is reduced from $\rm{SO}(1,3)^{\uparrow}$ to a 2-dimensional subgroup given by 
$\R^{\times}\times\rm{U}(1)$ ($\cong\mbb{C}^{\times}$), whose action on a null tetrad is
\begin{equation}\label{GHP}
 \ell^a\to a\ell^{a}, \qquad n^a\to a^{-1}n^{a}, \qquad  m^a\to z m^{a}, 
 \qquad  \bar{m}^a\to \bar{z}\bar{m}^{a},
\end{equation}
for $a\in\R^{\times}$ and $z\in\rm{U}(1)$.
Spinors, tensors, equations, etc. are projected on the null frame, so that the components $\eta$ of an arbitrary
spinor/tensor field get changed under a GHP transformation (\ref{GHP}), in the form $\eta\mapsto\l^p\bar{\l}^q\eta$ 
for some $p,q\in\mbb{Z}$ (where $\l=(az)^{1/2}$); i.e. they form representations of the GHP group.
These quantities are known as GHP weighted quantities of type $\{p,q\}$,
and can be thought of as sections of vector bundles $\mbb{E}_{\{p,q\}}$ associated to the GHP representations.
The covariant derivative on $\mbb{E}_{\{p,q\}}$ is 
\begin{equation}
\T_{a}=\p_{a}+p\w_{a}+q\bar\w_{a}, 
\end{equation}
where $\w_{a}$ is the GHP connection
\begin{equation}
 \w_{a}=-\e n_{a}+\e'\ell_{a}-\b' m_{a}+\b\bar{m}_{a}.
\end{equation}

Consider now the conformal manifold $(\M,[g_{ab}])$, which is a structure weaker than the spacetime.
We wish to use the preferred Weyl connection (\ref{A}) to construct a conformally covariant GHP
formalism. Let $\pi_P:P\to\M$ be a principal fiber bundle over the conformal manifold,
whose typical fiber over $x\in\M$ is the set of basis $\{N^a_{\bf a}\}$ of $T_x\M$ such that for 
any $\wh{g}_{ab}\in[g_{ab}]$, there exists $\a>0$ such that $\wh{g}_{ab}N^a_{\bf a}N^b_{\bf b}=\a^2\n_{\bf ab}$,
where $\n_{\bf ab}$ is defined by $\n_{01}=1=-\n_{23}$, with all other components being zero.
The structure group of $P$ is $G=\rm{SO}(1,3)^{\uparrow}\times\R^{+}$ (see e.g. \cite{Friedrich1977}). 
The Weyl connection $\slashed{\c}_a$ on $(\M,[g_{ab}])$ allows to define a notion of parallel transport 
on $\M$, which in turn induces a (local) connection 1-form on $P$ given by 
$\slashed{\w}{}_{a}{}^{\bf b}{}_{\bf c}=N^{\bf b}_{b}\slashed{\c}_{a}N^{b}_{\bf c}$ 
(where $N^{\bf a}_a$ is the frame dual to $N^a_{\bf a}$).
Note that, since (\ref{A}) is generally complex, $\slashed{\w}{}_{a}{}^{\bf b}{}_{\bf c}$ will take values 
in the complexification of the Lie algebra $\mf{g}={\rm Lie}(G)$, namely 
$\slashed{\w}\in T^{*}\M\otimes\mf{g}^{\mbb{C}}$, with $\mf{g}^{\mbb{C}}:=\mf{g}\otimes\mbb{C}$.
Now, in the GHP formalism one fixes a pair of null directions $\ell^a,n^a$. The subgroup of $G$ that 
preserves these null directions in the conformal structure is 
$G_o=\R^{\times}\times{\rm U}(1)\times\R^{+}$, whose action on the frame $\{N^a_{\bf a}\}$ is now
\begin{equation}\label{tnt}
 \ell^a\to\Om^{2w_0}a\ell^{a}, \qquad n^a\to\Om^{2w_1}a^{-1}n^{a}, \qquad
 m^a\to\Om^{w_0+w_1}z m^{a}, \qquad \bar{m}^a\to\Om^{w_0+w_1}\bar{z}\bar{m}^{a},
\end{equation}
where $a\in\R^{\times}$, $z\in{\rm U}(1)$, $\Om\in\R^{+}$, and $w_0,w_1$ are two real 
numbers that satisfy $w_0+w_1+1=0$ (in accordance to (\ref{ctsf})). The reduction $G\to G_o$ 
gives in turn a reduction of $P$ to another principal bundle $B$ with structure group $G_o$, 
and the induced connection form on $B$ is obtained as the parts of $\slashed{\w}{}_{a}{}^{\bf b}{}_{\bf c}$ 
that do not transform covariantly under $G_o$. It is straightforward to check that these 
parts are $\slashed{\w}{}_{a}{}^{0}{}_{0}, \slashed{\w}{}_{a}{}^{1}{}_{1}, \slashed{\w}{}_{a}{}^{2}{}_{2}$ 
and $\slashed{\w}{}_{a}{}^{3}{}_{3}$, which transform as
\begin{align}
 & \slashed{\w}{}_{a}{}^{0}{}_{0} \to \slashed{\w}{}_{a}{}^{0}{}_{0}+2w_0\Om^{-1}\slashed{\c}_a\Om+a^{-1}\slashed{\c}_{a}a, \\
 & \slashed{\w}{}_{a}{}^{1}{}_{1} \to \slashed{\w}{}_{a}{}^{1}{}_{1}+2w_1\Om^{-1}\slashed{\c}_a\Om+a\slashed{\c}_{a}a^{-1}, \\ 
 & \slashed{\w}{}_{a}{}^{2}{}_{2} \to \slashed{\w}{}_{a}{}^{2}{}_{2}+\Om^{-1}\slashed{\c}_a\Om+z^{-1}\slashed{\c}_{a}z, \\ 
 & \slashed{\w}{}_{a}{}^{3}{}_{3} \to \slashed{\w}{}_{a}{}^{3}{}_{3}+\Om^{-1}\slashed{\c}_a\Om+z\slashed{\c}_{a}z^{-1}.
\end{align}
By looking at these transformation behaviors, we can isolate the parts that transform only under each separate 
subgroup in the product $G_o=\R^{\times}\times{\rm U}(1)\times\R^{+}$; these are
$w_1\slashed{\w}{}_{a}{}^{0}{}_{0}-w_0\slashed{\w}{}_{a}{}^{1}{}_{1}=:K_a$ for $\R^{\times}$,
$\frac{1}{2}(\slashed{\w}{}_{a}{}^{3}{}_{3}-\slashed{\w}{}_{a}{}^{2}{}_{2})=:L_a$ for ${\rm U}(1)$, and 
$\frac{1}{2}(\slashed{\w}{}_{a}{}^{2}{}_{2}+\slashed{\w}{}_{a}{}^{3}{}_{3})=f_a$ for $\R^{+}$, 
since we have
\begin{equation}
 K_a\to K_a+a\slashed{\c}_a a^{-1}, \qquad L_a\to L_a+z\slashed{\c}_az^{-1}, 
 \qquad f_a\to f_a+\Om\slashed{\c}_a\Om^{-1}.
\end{equation}
The connection form on $B$ is then $\psi_a=(K_a,L_a,f_a)$, which takes values in the complexified Lie algebra
$\mf{g}^{\mbb{C}}_o=(\R\oplus\mf{u}(1)\oplus\R)\otimes\mbb{C}\cong\mbb{C}\oplus\mbb{C}\oplus\mbb{C}$.

Consider now a conformal tensor density, and project it over the frame $\{N^a_{\bf a}\}$ and its dual. 
These components (denoted generically by $\eta$) will be appropriately rescaled under (\ref{tnt}); i.e., they
form representations $\Pi_{b,s,w}:G_o\to{\rm GL}(\mbb{C})$ of $G_o$ given by
\begin{equation}\label{rep-bsw}
 \Pi_{b,s,w}(a,z,\Om)\eta:=a^b z^s \Om^w \eta, \qquad \eta\in\mbb{C}, \qquad (a,z,\Om)\in G_o,
\end{equation}
where $b$ and $s$ are integer numbers and $w\in\R$.
The components $\eta$ are then said to have {\em boost weight} $b$, {\em spin weight} $s$, and 
{\em conformal weight} $w$. In the spinor approach it is more natural to use instead the alternative 
$\{p,q\}$ weights defined by $b=(p+q)/2$, $s=(p-q)/2$, in which case we can denote the representation 
(\ref{rep-bsw}) as $\Pi_{p,q,w}$. These quantities will be said to have type $\{w;p,q\}$. The associated 
representation of the Lie algebra is 
\begin{equation}
 \pi_{p,q,w}(x,y,v)\eta=(p\tfrac{(x+y)}{2}+q\tfrac{(x-y)}{2}+wv)\eta.
\end{equation}
For (conformal) tensor {\em fields}, the components are sections of the associated vector bundles
\begin{equation}\label{avb}
 \mbb{E}_{\{p,q\}}[w]:=B\times_{\Pi_{p,q,w}}\mbb{C}.
\end{equation}
A covariant derivative on this structure can be induced by the connection form $\psi_a$ on $B$:
\begin{lem}
 The covariant derivative on the associated vector bundles $\mbb{E}_{\{p,q\}}[w]$ is
 \begin{equation}\label{icd}
  \slashed{\T}_a\eta=\p_a\eta+\pi_{p,q,w}(\psi_a)\eta.
 \end{equation}
 Explicitly, this is
 \begin{equation}\label{slashT}
  \slashed{\T}_a\eta=\p_a\eta +wf_a\eta +p(\w_a+B_a)\eta +q(\bar{\w}_a+C_a)\eta,
 \end{equation}
 where $\w_a$ is the GHP connection form, and the complex 1-forms $B_a$, $C_a$ are defined by
 \begin{align}
 B_{a}&:=\tfrac{1}{2}(K_a+L_a)=w_1(\rho n_{a}-\t\bar{m}_{a})-w_0(\rho' \ell_{a}-\t' m_{a}), \label{B}  \\
 C_{a}&:=\tfrac{1}{2}(K_a-L_a)=w_1(\rho n_{a}-\t' m_{a})-w_0(\rho' \ell_{a}-\t\bar{m}_{a}). \label{C}
\end{align}
\end{lem}

\begin{proof}
This is deduced from our discussion above, and from the general formula for the covariant derivative 
induced on associated vector bundles.
\end{proof}

\begin{remark}[{\bf The Teukolsky connection}]
Choose the conformal weights $w_0=0$ and $w_1=-1$ in (\ref{ctsf}). 
For quantities with $w=0$ and $q=0$, $\slashed{\T}_a$ {\em coincides exactly} with the Teukolsky 
derivative $D_a$. (Note that, since the fields we are interested in --namely $\vp_{A_1...A_{2\s}}$-- 
have no primed indices, no $q$-weight will appear in our general formulae.) 
For conformal weight $w\neq0$, by writing $f_a$ in 
the form (\ref{Apsi2}) we see that the covariant derivative (\ref{slashT}) describes exactly the 
combination of Teukolsky derivatives and different powers of $\Psi^{1/3}_2$ that appear in formulae
involving $\iota^A$, see section \ref{sec-rt}.
\end{remark}

From the discussion above we can conclude: 
\begin{thm}\label{thm-teuk}
The Teukolsky connection is originated in the covariant derivative naturally induced on the vector bundles 
(\ref{avb}) associated to representations of the GHP-conformal symmetry.
\end{thm}

An important property of (\ref{icd}) is the following:
\begin{prop}
$\slashed{\T}_a$ commutes with the GHP prime operation\footnote{We recall that the GHP prime operation 
is defined as the transformation $o^A\to i\iota^A$, $\iota^A\to io^A$.}:
\begin{equation}\label{prime}
(\slashed{\T}_{a}\eta)'=\slashed{\T}_{a}\eta'.
\end{equation}
\end{prop}

\begin{proof}
If $\eta$ is of type $\{w;p,q\}$, then $\eta'$ is of type $\{w';-p,-q\}$, where $w'$ is given by 
$w'= w-(w_0-w_1)(p+q)$. Using that
\begin{align}
 & B'_{a}=-(w_0-w_1)f_{a}-B_{a}, \\
 & C'_{a}=-(w_0-w_1)f_{a}-C_{a},
\end{align}
the result follows easily by calculating both sides of (\ref{prime}).
\end{proof}

As in the usual GHP formalism, this result allows one to halve the number of calculations.
On the other hand, $\slashed{\T}_a$ {\em does not} commute with complex conjugation. 
This is related to our {\em chiral} treatment, see Remark \ref{remarkCh}.

For spinor/tensor fields that are GHP and conformally weighted and have a nontrivial structure of indices,
we must complement (\ref{slashT}) with 
the objects (\ref{Qcc}) and (\ref{Wcc}) in order to get a derivative operator that is covariant under 
both kinds of transformations.
To this end, we define a covariant derivative $\C_a$ that acts on a vector $v^b$ and a spinor $\k^B$,
both with conformal weight $w$ and GHP type $\{p,q\}$, as
\begin{align}
 & \C_{a}v^b = \slashed{\c}_{a}v^{b}+[wf_a+p(\w_a+B_a)+q(\bar\w_a+C_a)]v^b, \label{Cv} \\
 & \C_{a}\k^B = \slashed{\c}_a\k^B+[wf_a+p(\w_a+B_a)+q(\bar\w_a+C_a)]\k^B. \label{Cs}
\end{align}
The extension of the action of $\C_{a}$ to a spinor/tensor field with an arbitrary structure of indices 
follows straightforwardly.
One can then check that, under conformal and GHP transformations, for an arbitrary spinor/tensor field $\Psi$
of type $\{w;p,q\}$ we have
\begin{equation}
 \C_{a}\Psi \; \to \;  \Om^w \l^p \bar\l^q \C_{a}\Psi.
\end{equation}

Finally, the projection of $\C_{a}$ on a null tetrad defines the conformally weighted GHP operators
\begin{equation}
 \tho_{\C}:=\ell^{a}\C_{a}, \qquad  \tho'_{\C}:=n^{a}\C_{a}, \qquad
  \edt_{\C}:=m^{a}\C_{a}, \qquad  \edt'_{\C}:=\bar{m}^{a}\C_{a}.
\end{equation}
We emphasize that, with the above definition of $\C_a$, these operators act on arbitrary
{\em spinor/tensor} conformal densities.
For the particular case of a {\em scalar} conformal density $\eta$ of type $\{w;p,q\}$, we have
\begin{align}
 \tho_{\C}\eta&=[\tho+(w+(p+q)w_1)\rho]\eta, \\
 \tho'_{\C}\eta&=[\tho'+(w-(p+q)w_0)\rho']\eta, \\
 \edt_{\C}\eta&=[\edt+(w+pw_1-qw_0)\t]\eta, \\
 \edt'_{\C}\eta&=[\edt'+(w-pw_0+qw_1)\t']\eta,
\end{align}
this way we recover the operators defined by Penrose $\&$ Rindler in \cite[section 5.6]{Penrose1} 
(see Eqs. (5.6.36) in that reference).

\subsection{Parallel and Killing spinors}\label{sec-wks}

As we have seen, the almost-comlpex structure (\ref{acs}) of type D spacetimes determines 
a unique Weyl connection on the conformal structure, and thus we have a naturally induced covariant 
derivative on weighted spinor/tensor bundles, given by (\ref{Cs}), (\ref{Cv}), (\ref{icd}).
This leads us to a geometrically meaningful generalization of some interesting differential 
equations; for example, we are naturally led to the notion of a {\em weighted Killing spinor} 
as an element $\w^{A_1...A_n}$ of $\mbb{S}^{A_1...A_n}_{\{p,q\}}[w]$ 
that satisfies the equation
\begin{equation}\label{wks}
 \C_{A'}{}^{(A}\w^{B_1...B_n)}=0.
\end{equation}
A stronger condition is that of a {\em parallel spinor}, i.e. a solution of 
\begin{equation}
 \C_{a}\w^{B_1...B_n}=0.
\end{equation}

\begin{lem}\label{lem-ps}
Let $(\M,[g_{ab}],J)$ be the conformal almost-Hermitian manifold associated to an Einstein spacetime 
$(\M,g_{ab})$, where $J$ is given by (\ref{acs}) for an arbitrarily chosen pair of null directions $\{o^A,\iota^A\}$, 
with $o^A \in\mbb{S}^{A}_{\{1,0\}}[w_0]$, $\iota^A \in\mbb{S}^{A}_{\{-1,0\}}[w_1]$ and $o_A\iota^A=1$. 
Then the following two are equivalent:
\begin{itemize}
\item[$(i)$] The spacetime is algebraically special, with principal null direction $o^A$.
\item[$(ii)$] The spinor $o^A$ is parallel with respect to the naturally induced covariant derivative, 
namely $\C_ao^B=0$.
\end{itemize}
In particular, the spacetime is of Petrov type D if and only if both spinors are parallel, $\C_ao^B=0=\C_a\iota^B$.
\end{lem}

\begin{proof}
A straightforward calculation shows that (for any conformal weights $w_0,w_1$ in (\ref{ctsf}))
\begin{equation}\label{Co}
 \C_{a}o^B=(-\k n_{a}+\sigma\bar{m}_{a})\iota^B.
\end{equation}
By the Goldberg-Sachs theorem, algebraic speciality is equivalent to $\k=\sigma=0$, thus the 
result $(i)\Leftrightarrow(ii)$ follows. Since $\C_a$ commutes with the prime operation, we immediately see that 
we also have $\C_{a}\iota^B=0$ for a type D spacetime, and that, conversely, if $\C_ao^B=0=\C_a\iota^B$, 
then the spacetime is type D.
\end{proof}

Note that, since $\C_a$ does not commute with complex conjugation, the condition $\C_a o^B=0$ {\em does not} 
imply $\C_a\bar{o}^{B'}=0$; in fact we generally have $\C_a\bar{o}^{B'}\neq0$.

The result of Lemma \ref{lem-ps} then explains the modified Killing spinor equation (\ref{mks}). 
For, in a type D spacetime, by Remark \ref{remarkA} 
we can write the 1-form $f_a$ as in (\ref{Apsi2}), and recalling the definition of the Teukolsky derivative given 
in the introduction (using now the more general form (\ref{B}) for $B_a$), after an easy calculation we get
\begin{align}
 & \C_{A'}{}^{(A}o^{B)}=\Psi^{-w_0/3}_2 D_{A'}{}^{(A}[\Psi^{w_0/3}_2o^{B)}], \\
 & \C_{A'}{}^{(A}\iota^{B)}=\Psi^{-w_1/3}_2 D_{A'}{}^{(A}[\Psi^{w_1/3}_2\iota^{B)}].
\end{align}
Choosing $w_0=0$, $w_1=-1$, we immediately obtain (\ref{mks}) on a type D spacetime, as a 
consequence of Lemma \ref{lem-ps}.
Note that, in particular, the vanishing of $\C_a[o^B\iota^C]$ gives the ordinary Killing spinor of 
type D spaces mentioned in the introduction.

\subsection{Weighted local twistors}\label{sec-wlt}

A useful framework for dealing with conformal geometry in arbitrary dimensions is the {\em tractor formalism}, 
see e.g. \cite{Curry2014}. When specialized to conformal spinor geometry in 4 dimensions, the tractor 
calculus becomes the {\em local twistor} formalism, which is a possible generalization of the original twistor 
theory to curved spacetimes, see \cite{PenroseMacCallum, Dighton} and \cite[section 6.9]{Penrose2}. 
In this section we show that the distinguished Weyl connection of conformal type D spaces leads 
to a natural injection of spinor space into the 
local twistor bundle, and make some general remarks about the construction.

A local twistor can be represented as a pair of spinors $(\w^A,\pi_{A'})$ such that under a conformal 
transformation, the local twistor itself is invariant, but its representation into spinor parts is changed 
according to $\wh{\w}^A=\w^A$, $\wh{\pi}_{A'}=\pi_{A'}+i\U_{AA'}\w^A$. This is usually described by 
the short exact sequence
\begin{equation}\label{es1}
 0 \to \mbb{S}_{A'} \to \mbb{T}^{\a} \to \mbb{S}^{A} \to 0,
\end{equation}
where $\mbb{T}^{\a}$ is the space of local twistors\footnote{Greek letters $\a,\b,\g,...$ denote twistor 
indices, and take values in $\{0,1,2,3\}$.}, and the second and third maps 
are given respectively by $\pi_{A'}\mapsto (0,\pi_{A'})$ and $(\w^A,\pi_{A'})\mapsto \w^A$. 
The sequence is then conformally invariant.
The space $\mbb{T}^{\a}$ is a vector bundle over $\M$ with structure group $\rm{SU}(2,2)$, and 
with a conformally invariant connection called {\em local twistor transport}, given by
\begin{equation}\label{cdt}
 {}^{\text{\sc t}}\c_{a}\Z^{\b} := 
 ( \c_{a}\w^B+i\e_{A}{}^{B}\pi_{A'} , \c_{a}\pi_{B'}+iP_{AA'CB'}\w^C),
\end{equation}
where $P_{AA'BB'}$ is the Schouten tensor (\ref{schouten}). 
A {\em global twistor} is one which is parallel under local twistor transport, and in that case it coincides 
with the usual twistor concept in a (conformally) flat spacetime. 
The {\em primary spinor part} (i.e. the spinor part with all its indices at the upper position) of a 
global twistor is a Killing spinor: $\c_{A'A}\w^B=-i\e_{A}{}^{B}\pi_{A'}$.

An exact sequence like (\ref{es1}) but with the arrows in the opposite direction can be obtained 
by means of the 1-form (\ref{A}). This is because (\ref{A}) allows a natural injection of $\mbb{S}^A$ into 
local twistor space, by mapping a spinor field $\w^A$ (with conformal weight zero) to $(\w^A,\a_{A'})$, 
where $\a_{A'}=-if_{AA'}\w^A$. In other words, we have the short exact sequence
\begin{equation}\label{es2}
 0 \to \mbb{S}^{A} \to \mbb{T}^{\a} \to \mbb{S}_{A'} \to 0,
\end{equation}
where the second and third maps are now $\w^A\mapsto(\w^A,\a_{A'})$ and 
$(\w^A,\pi_{A'}) \mapsto \pi_{A'}+if_{AA'}\w^A$. The exactness of the sequence and its conformal 
invariance are  easily checked. Consider for example a spin frame $\{o^A,\iota^A\}$, and choose 
the conformal weights $w_0=0$ and $w_1=-1$ in (\ref{ctsf}). 
Via (\ref{es2}), we have the local twistor $\X^{\a}=(o^A,\a_{A'})$.
On the other hand, the pair $\Y^{\a}=(\iota^A,\b_{A'})$, where $\b_{A'}:=-if_{AA'}\iota^A$, is 
a conformally weighted local twistor, with conformal weight $w=-1$, i.e. a section of 
$\mbb{T}^{\a}\otimes\mc{E}[-1]$.
But note that $\X^{\a}$ and $\Y^{\a}$ are also GHP weighted: they have GHP types $\{1,0\}$
and $\{-1,0\}$ respectively. This situation then leads us naturally to the 
consideration of {\em weighted} local twistors.

We will say that $\Z^{\a}$ is a weighted local twistor if, under conformal and GHP transformations, 
it transforms as $\Z^{\a}\to \Om^{w}\Z^{\a}$ and $\Z^{\a}\to \l^p\bar\l^q \Z^{\a}$ 
respectively. In the spinor representation $\Z^{\a}=(\w^A,\pi_{A'})$, both spinor parts have
GHP type $\{p,q\}$, and their conformal behavior is $\wh{\w}^A=\Om^{w}\w^A$, 
$\wh{\pi}_{A'}=\Om^{w}(\pi_{A'}+i\U_{AA'}\w^A)$.
The set of weighted local twistors can then be identified with the vector bundle 
$\mbb{T}^{\a}_{\{p,q\}}[w]:=\mbb{T}^{\a}\otimes\mbb{E}_{\{p,q\}}[w]$.
We can construct a connection on this structure by combining the usual local twistor transport (\ref{cdt})
with the covariant derivative (\ref{slashT}), extended to act `trivially' on spinor indices. 
More precisely, we define
\begin{equation}\label{cdwt}
 {}^{\text{\sc t}}\C_{a}\Z^{\b} := 
 ( \slashed{\T}_{a}\w^B+i\e_{A}{}^{B}\pi_{A'} , \slashed{\T}_{a}\pi_{B'}+iP_{AA'CB'}\w^C).
\end{equation}
In this expression, we have\footnote{the spinor 
$\w^A$ and the GHP connection form $\w_a$ should not be confused in this formula.} 
$\slashed{\T}_a\w^B=\c_a\w^B+wf_a\w^B+p(\w_a+B_a)\w^B+q(\bar{\w}_a+C_a)\w^B$, and analogously 
for $\slashed{\T}_a\pi_{B'}$.
One can then check that ${}^{\text{\sc t}}\C_{a}\Z^{\b}$ is again a weighted local twistor, with 
conformal weight $w$ and GHP type $\{p,q\}$; therefore (\ref{cdwt}) gives a connection on 
$\mbb{T}^{\a}_{\{p,q\}}[w]$.
In this framework we can define a weighted {\em global} twistor as one which is parallel under 
weighted local twistor transport (\ref{cdwt}).
This leads us to the definition of a weighted Killing spinor as the primary spinor part of a parallel 
weighted global twistor:
\begin{equation}
 \slashed{\T}_{A'A}\w^B=-i\e_{A}{}^{B}\pi_{A'}.
\end{equation}
The generalization of this equation is
\begin{equation}
 \slashed{\T}_{A'}{}^{(A}\w^{B_1...B_n)}=0.
\end{equation}
It is easily checked that $\slashed{\T}_{A'}{}^{(A}\w^{B_1...B_n)}= \C_{A'}{}^{(A}\w^{B_1...B_n)}$, therefore 
this definition coincides with (\ref{wks}) (the symmetrization being crucial for this to hold).

\section{Massless fields}

\subsection{Conformally covariant identities}\label{sec-cci}

We now wish to find conformally covariant identities for massless free fields propagating in a 
background curved spacetime, and their possible relation to the identities (\ref{mainid0}).

Let $\vp_{A_1...A_n}$ be an arbitrary symmetric spinor field with conformal weight $-1$ and GHP type $\{0,0\}$.
Recalling the definition of the conformally covariant derivative (\ref{ccd}), for any 1-form $f_a$ associated to 
a Weyl connection we have
\begin{equation}\label{cfeq}
\c^{A_1A'_1}\vp_{A_1...A_n}=\C^{A_1A'_1}\vp_{A_1...A_n}.
\end{equation}
From now on we will use the 1-form $f_a$ given by (\ref{A}), since, as we have seen, it is 
naturally distinguished for the algebraically special spacetimes we are interested in, in relation 
to their almost-complex structure.
The simplest operation that we can effect on $\C^{A_1A'_1}\vp_{A_1...A_n}$ to get a 
conformally covariant second order differential operator, with $n$ totally symmetric unprimed indices,
is to simply take an additional covariant derivative:
\begin{equation}\label{mainspinorop}
 \C_{A'_1(A_1}\C^{A'_1B}\vp_{A_2...A_n)B}.
\end{equation}
It is tedious but straightforward to show that, in an arbitrary spacetime, we have
\begin{align}
\nonumber 2\C_{A'_1(A_1}\C^{A'_1B}\vp_{A_2...A_n)B} =& (\Box_{\{-1;0,0\}}-(n+2)\Psi_2+n\zeta)\vp_{A_1...A_n}\\
 &  -n\m_{(A_1}{}^{B}\vp_{A_2...A_n)B} -2(n-1)\Psi_{(A_1A_2}{}^{BC}\vp_{A_3...A_n)BC},
\end{align}
where we have defined the conformally and GHP covariant wave operator
\begin{equation}\label{ccwo}
 \Box_{\{w;p,q\}}:=g^{ab}\C_a\C_b
\end{equation}
acting on $\mbb{S}_{\{p,q\}A_1...A_n}[w]$,
and also $\zeta:= \sigma\sigma'-\k\k'$, $\m_{AB}:= \chi \iota_A\iota_B- \chi'o_Ao_B$, and 
$\chi := (\tho'+2\rho'-\bar\rho')\k-(\edt'+2\t'-\bar\t)\sigma+2\Psi_1$.
We want to obtain equations for the components of $\vp_{A_1...A_n}$ in an arbitrary spin dyad
$\{o^A,\iota^A\}$, where we choose the conformal weights $w_0=0$ and $w_1=-1$ (with respect to (\ref{ctsf})).
The components are defined by
\begin{equation}
 \vp_k:=\vp_{A_1...A_{k}A_{k+1}...A_n}\iota^{A_1}...\iota^{A_k}o^{A_{k+1}}...o^{A_n}.
\end{equation}
The conformal weight of $\vp_k$ is $w=-1-k$, and its GHP type is $\{n-2k,0\}$. 
Consider the extreme spin weight case; in particular the component $\vp_0$.
Defining $o^{A_1...A_n}=o^{A_1}...o^{A_n}$ and projecting on $o^{A_1...A_n}$, after some tedious calculations
(using (\ref{Co})) we find
\begin{equation}\label{idc0}
 2 o^{A_1...A_n}\C_{A'_1(A_1}\C^{A'_1B}\vp_{A_2...A_n)B}=(\Box_{\{-1;n,0\}}-3n\Psi_2)\vp_{0}+F_n[\vp_1]+G_n[\vp_2],
\end{equation}
where $F_n[\vp_1] = -2n(-\k\tho'_{\C}+\sigma\edt'_{\C}-\chi-\Psi_1)\vp_1-4\Psi_1\vp_1$ 
and $G_n[\vp_2] = -2(n-1)\Psi_0\vp_2$.
Imposing now the field equations $\C^{A'_1A_1}\vp_{A_1...A_n}=0$, and assuming the components 
$\vp_1$ and $\vp_2$ to be arbitrary, we see that $\vp_0$ satisfies a decoupled equation if and only if
$\k=\sigma=\Psi_0=\Psi_1=0$, that is, {\em if and only if the spacetime is algebraically special},
with $o^A$ aligned to the principal null direction. The decoupled equation in such case is
\begin{equation}
 (\Box_{\{-1;n,0\}}-3n\Psi_2)\vp_{0}=0.
\end{equation}

\begin{remark}
Note that the condition $\k=\sigma=\Psi_0=\Psi_1=0$ imposed on the spacetime is well defined on the 
conformal class, since both the spin coefficients $\k, \sigma$ and the Weyl scalars are conformal 
densities, thus the condition extends to all metrics in the conformal class.
This is just another way of saying that the algebraic speciality of a spacetime is common 
to the whole conformal class.
\end{remark}

On the other hand, since the covariant derivative $\C_a$ commutes with the prime operation, 
from (\ref{idc0}) we immediately deduce an analogous identity for the component with 
opposite extreme spin weight:
\begin{equation}\label{idc1}
 2 \iota^{A_1...A_n}\C_{A'_1(A_1}\C^{A'_1B}\vp_{A_2...A_n)B}=(\Box_{\{-n-1;-n,0\}}-3n\Psi_2)\vp_{n}+F'_n[\vp_{n-1}]+G'_n[\vp_{n-2}].
\end{equation}
We then deduce that {\em both} extreme components, $\vp_0$ and $\vp_n$, satisfy decoupled equations if and 
only if it holds $\k=\sigma=\k'=\sigma'=\Psi_0=\Psi_1=\Psi_3=\Psi_4=0$, namely, {\em if and only if the spacetime
is of Petrov type D}.

The previous discussion is valid for fields with arbitrary spin. We now particularize to the 
physically relevant cases of spin $1/2$, 1 and 2.

\subsubsection{Spin $\s=1/2$}

The results just described apply directly. In particular, we have the following identities valid 
for an arbitrary spacetime:
\begin{align}
 & 2 \; o^{A}\C_{A'A}\C^{A'B}\vp_{B}=(\Box_{\{-1;1,0\}}-3\Psi_2)\vp_{0}+F_1[\vp_1], \\
 & 2 \; \iota^{A}\C_{A'A}\C^{A'B}\vp_{B}=(\Box_{\{-2;-1,0\}}-3\Psi_2)\vp_{1}+F'_1[\vp_0].
\end{align}
The extreme component $\vp_0=o^A\vp_A$ of an arbitrary Weyl-Dirac field satisfies a decoupled wave-like 
equation if and only if the spacetime is algebraically special, with PND aligned to $o^A$.
Both components $\vp_0$ and $\vp_1$ decouple if and only if the spacetime is of Petrov type D.

\subsubsection{Spin $\s=1$}

For $\s=1$ the previous results also apply directly. Besides the extreme weight cases, we have 
to consider the corresponding identity for spin weight zero; this can be calculated 
analogously to the others. For an arbitrary spacetime, we get the following identities:
\begin{align}
 & 2 \; o^{AB}\C_{A'(A}\C^{A'C}\vp_{B)C}=(\Box_{\{-1;2,0\}}-6\Psi_2)\vp_{0}+F_2[\vp_1]+G_2[\vp_2], \\
 & 2 \; o^{A}\iota^{B}\C_{A'(A}\C^{A'C}\vp_{B)C}=(\Box_{\{-2;0,0\}}-2\zeta)\vp_{1}+H_2[\vp_2]+H'_2[\vp_0], \\
 & 2 \; \iota^{AB}\C_{A'(A}\C^{A'C}\vp_{B)C}=(\Box_{\{-3;-2,0\}}-6\Psi_2)\vp_{2}+F'_2[\vp_1]+G'_2[\vp_0],
\end{align}
where we have defined $H_{2}[\vp_2]:=2(\k\tho'_{\C}-\sigma\edt'_{\C}+\chi-2\Psi_1)\vp_2$.
As in the spin $\s=1/2$ case, the extreme component of a generic Maxwell field decouples if and only 
if the spacetime is algebraically special.
The spin weight zero component, on the other hand, decouples if and only if the spacetime is type D.

\subsubsection{Spin $\s=2$}\label{sec-spin2}

This case is much more subtle than the previous ones, since if we intend to describe ``massless free fields 
of spin 2 propagating in a curved background spacetime'' as perturbations of the Weyl curvature spinor,
then the conformally covariant description that we developed is not appropriate, since $\Psi_{ABCD}$ 
has conformal weight $w=0$, and it should have $w=-1$ in order for our general formulae to be valid.
This suggests that we describe the required massless spin 2 fields in terms of a {\em rescaled Weyl spinor}, 
defined in the following way.
Given the conformal class $[g_{ab}]$ associated to $g_{ab}$, for an arbitrary representative 
$[g_{ab}]\ni\wh{g}_{ab}=\Om^2 g_{ab}$ we introduce
\begin{equation}\label{rws}
 \vp_{ABCD}:=\Om^{-1}\Psi_{ABCD}.
\end{equation}
This object has conformal weight $w=-1$, and it is actually very used in studies of conformal and asymptotic 
aspects of the Einstein field equations and in the analysis of conformal infinity in General Relativity, see e.g. 
\cite{Friedrich2}, \cite[sections 9.6, 9.7]{Penrose2} and \cite[Ch. 8 and 10]{ValienteKroon}. 
In fact, following closely Penrose $\&$ Rindler \cite{Penrose2}, we can think of (\ref{rws}) as the 
``gravitational spin 2 field'' (see discussion around eq. (9.6.40) and theorem (9.6.41) in \cite{Penrose2}), 
in the sense that the property of {\em peeling} allows to interpret the components of (\ref{rws}) as 
describing the gravitational radiation field near conformal infinity in asymptotically simple 
spacetimes (\cite[eq. (9.7.38)]{Penrose2}).

Let $(\M,g_{ab}(\ve))$ be a monoparametric family of spacetimes, such that $g_{ab}(0)$ satisfies the 
vacuum Einstein equations with cosmological constant. Consider the conformal class associated to 
$g_{ab}(\ve)$, denoted as $[g_{ab}(\ve)]$, where $\wh{g}_{ab}(\ve)\in[g_{ab}(\ve)]$
if and only if there exists $\Om(\ve)>0$ such that $\wh{g}_{ab}(\ve)=\Om^2(\ve)g_{ab}(\ve)$. 
Note that we impose the conformal factor to depend on the parameter $\ve$; in particular, we take 
$\Om(0)\equiv \mr\Om=\rm{constant}$ (the reason for this can be inferred from the calculations given 
in appendix \ref{appendix}).
For an arbitrary representative in $[g_{ab}(\ve)]$, we have $\vp_{ABCD}(\ve)=\Om^{-1}(\ve)\Psi_{ABCD}(\ve)$.
Linearizing equation (\ref{idc0}) for $n=4$, we find\footnote{linearized quantities are denoted with a 
dot, e.g. $\dot{T}:=\frac{d}{d\ve}|_{\ve=0}[T(\ve)]$; quantities without a dot in the right hand side of 
(\ref{lce0}) are understood as evaluated in the background.}
\begin{equation}\label{lce0}
 \tfrac{d}{d\ve}|_{\ve=0}\{2\;o^{ABCD}\C_{A'(A}\C^{A'E}\vp_{BCD)E}\}
 =(\Box_{\{-1;4,0\}}-12\Psi_2)\dot\vp_0-6\vp_2\dot\Psi_0+\dot{B}_0,
\end{equation}
where
\begin{equation}
 \dot{B}_0=(\dot\Box_{\{-1;4,0\}}-12\dot\Psi_2)\vp_0-6\Psi_0\dot\vp_2+\dot{F}_4.
\end{equation}
Now, the field equations describing gravitational perturbations are the linearized vacuum (with cosmological 
constant) Einstein equations, which imply the linearized Bianchi identities
$\frac{d}{d\ve}|_{\ve=0}(\c^{A'A}\Psi_{ABCD})=0$.
Replacing $\vp_{ABCD}$ by $\Om^{-1}\Psi_{ABCD}$ in the left hand side of (\ref{lce0}) and imposing the 
just mentioned field equations, after some tedious calculations (see appendix \ref{appendix}) 
we get that $\dot\vp_0$ satisfies a decoupled equation if and only if $\k=\sigma=\Psi_0=\Psi_1=0$, 
that is, if and only if the background spacetime is algebraically special, with PND aligned to $o^A$.
Note that in such case we have $\dot\vp_0=\mr\Om^{-1}\dot\Psi_0$.
For the component with opposite extreme spin weight, namely $\dot\vp_4$, we have
\begin{equation}\label{lce4}
 \tfrac{d}{d\ve}|_{\ve=0}\{2\;\iota^{ABCD}\C_{A'(A}\C^{A'E}\vp_{BCD)E}\}
 =(\Box_{\{-5;-4,0\}}-12\Psi_2)\dot\vp_4-6\vp_2\dot\Psi_4+\dot{B}_4,
\end{equation}
where $\dot{B}_4=(\dot{B}_0)'$.
Imposing the field equations and using the explicit expression of $\dot{B}_4$,
we then see that {\em both} extreme components, $\dot\vp_0$ and $\dot\vp_4$, decouple if and only if
it holds $\k=\sigma=\k'=\sigma'=\Psi_0=\Psi_1=\Psi_3=\Psi_4=0$, i.e. if and only if the 
background spacetime is of Petrov type D. In that case we have $\dot\vp_0=\mr\Om^{-1}\dot\Psi_0$ 
and $\dot\vp_4=\mr\Om^{-1}\dot\Psi_4$.

Finally, consider the spin weight zero component $\vp_2$. A calculation similar to (\ref{idc0}) leads 
to the following result in an arbitrary spacetime:
\begin{equation}\label{ce2}
 2\;o^{AB}\iota^{CD}\C_{A'(A}\C^{A'E}\vp_{BCD)E}=\Box_{\{-3;0,0\}}\vp_2+B_2,
\end{equation}
where
\begin{align}
\nonumber B_2=&-(\k\tho'_{\C}-\sigma\edt'_{\C}-4\chi+16\Psi_1)\vp_3
  -(\k'\tho_{\C}-\sigma'\edt_{\C}-4\chi'+16\Psi_3)\vp_1\\
  & -8\zeta\vp_2-6(\Psi_0\vp_4+\Psi_4\vp_0).
\end{align}
Linearizing equation (\ref{ce2}), the only interesting case we find is when the background spacetime is type D,
where it holds
\begin{equation}\label{lce2}
 \tfrac{d}{d\ve}|_{\ve=0}\{2\;o^{AB}\iota^{CD}\C_{A'(A}\C^{A'E}\vp_{BCD)E}\}=
  \Box_{\{-3;0,0\}}\dot\vp_2+\dot\Box_{\{-3;0,0\}}\vp_2. 
\end{equation}
Now, it is important to note that, on shell, the left hand side of (\ref{lce2}) {\em does not vanish}, 
but it is equal to $\Box_{\{-3;0,0\}}[\Psi_2\dot\Om^{-1}]$, see appendix \ref{appendix}. 
We then see that, on shell, $\dot\vp_2$ does not decouple, since the terms $\Box_{\{-3;0,0\}}[\Psi_2\dot\Om^{-1}]$
and $\dot\Box_{\{-3;0,0\}}\vp_2$ do not vanish. 
On the other hand, $\Box_{\{-3;0,0\}}[\Psi_2\dot\Om^{-1}]$ cancels one of the terms appearing in the right 
hand side of (\ref{lce2}) (after replacing $\dot\vp_2=\Psi_2\dot\Om^{-1}+\mr\Om^{-1}\dot\Psi_2$), 
so that, on shell, we are left with the equation
\begin{equation}\label{eqlwe}
 \Box_{\{-3;0,0\}}[\mr\Om^{-1}\dot\Psi_2]+\dot\Box_{\{-3;0,0\}}\vp_2=0.
\end{equation}

Summarizing, the conformally covariant identities we find for the case in which the background spacetime 
is type D are:
\begin{align}
 \tfrac{d}{d\ve}|_{\ve=0}\{2\;o^{ABCD}\C_{A'(A}\C^{A'E}\vp_{BCD)E}\} = & (\Box_{\{-1;4,0\}}-18\Psi_2)\dot\vp_0, \label{id0D} \\
 \tfrac{d}{d\ve}|_{\ve=0}\{2\;o^{AB}\iota^{CD}\C_{A'(A}\C^{A'E}\vp_{BCD)E}\} = 
   &\Box_{\{-3;0,0\}}\dot\vp_2+\dot\Box_{\{-3;0,0\}}\vp_2, \label{id2D}\\
 \tfrac{d}{d\ve}|_{\ve=0}\{2\;\iota^{ABCD}\C_{A'(A}\C^{A'E}\vp_{BCD)E}\} = &(\Box_{\{-5;-4,0\}}-18\Psi_2)\dot\vp_4. \label{id4D}
\end{align}
The left hand sides of equations (\ref{id0D}) and (\ref{id4D}) vanish on shell, leaving us with decoupled 
equations for $\dot\vp_0$ and $\dot\vp_4$ respectively, where the normally hyperbolic operator 
$\Box_{\{w;p,q\}}$ (defined in (\ref{ccwo})) has a well defined geometrical meaning in terms of 
conformal and GHP covariance (see section \ref{sec-ccghp}).
The left hand side of (\ref{id2D}) does not vanish on shell, and the discussion between equations (\ref{lce2}) and 
(\ref{eqlwe}) applies.

\subsection{Relation with Teukolsky operators}\label{sec-rt}

In the previous subsection we found conformally covariant identities for fields with spin $1/2$, 1 and 2, 
and we also found the conditions that must satisfy the metrics of the conformal class associated to the 
background spacetimes in order for the field components to satisfy wave-like decoupled equations.
In order to relate these results with the main identity (\ref{mainid0}), we need the relation between 
the conformally covariant operators that we used, and the Teukolsky operators that appear in (\ref{mainid0}).
We have:
\begin{lem}\label{lem-we}
Consider an Einstein spacetime of Petrov type D, and let $\eta$ be a conformal scalar density of type $\{w; p,0\}$.
The relation between the conformally-GHP covariant wave operator $\Box_{\{w;p,0\}}$ and the Teukolsky 
operator $\teuk_{p}$ is given by:
\begin{equation}
 \Box_{\{w;p,0\}}\eta= \Psi_2^{-(w+1)/3}(\teuk_{p}+2\Psi_2+\tfrac{R}{6})(\Psi^{(w+1)/3}_2\eta).
\end{equation}
\end{lem}

\begin{proof}
Recall the definition of the Teukolsky derivative (acting on GHP type $\{p,0\}$ quantities), 
$D_a=\T_a+pB_a$ (where $B_a$ is given by (\ref{B})), and the associated wave operator $\teuk_{p}=D^aD_a$. 
Since the spacetime is type D we can express the 1-form (\ref{A}) as (\ref{Apsi2})\footnote{Recall that this way 
one is breaking the explicit conformal covariance of the formalism, see Remark \ref{remarkA}.}, 
then it is easy to show that, for an arbitrary number $z\in\R$, we have
\begin{equation}\label{teuk2}
  \Psi_2^{-z/3}\teuk_{p}(\Psi^{z/3}_2\eta)=\teuk_{p}\eta+2zf^aD_a\eta+z[\c_af^a+zf_af^a]\eta.
\end{equation}
On the other hand, for $\Box_{\{w;p,0\}}\eta$ we find
\begin{equation}
 \Box_{\{w;p,0\}}\eta=\teuk_{p}\eta+2(w+1)f^aD_a\eta+w[\c_af^a+(w+2)f_af^a]\eta.
\end{equation}
Using now (\ref{teuk2}) with $z=w+1$ and the identity $\c_af^a+f_af^a=-(2\Psi_2+\frac{R}{6})$, 
the result follows.
\end{proof}

Using Lemma \ref{lem-we}, we deduce immediately the following identities valid for perturbations of 
an Einstein type D spacetime (where $\l$ is the cosmological constant):\\

\noindent
Spin $\s=1/2$:
\begin{align}
 2 \; o^{B}\C_{A'B}\C^{A'A}\vp_{A}=& \;(\teuk_{+1}-\Psi_2+\tfrac{2}{3}\l)\vp_{0}, \\
 2 \; \iota^{B}\C_{A'B}\C^{A'A}\vp_{A}=& \;\Psi^{1/3}_2(\teuk_{-1}-\Psi_2+\tfrac{2}{3}\l)[\Psi^{-1/3}_2\vp_{1}].
\end{align}

\noindent
Spin $\s=1$:
\begin{align}
 2 \; o^{AB}\C_{A'(A}\C^{A'C}\vp_{B)C}=& \;(\teuk_{+2}-4\Psi_2+\tfrac{2}{3}\l)\vp_{0}, \\
 2 \; o^{A}\iota^{B}\C_{A'(A}\C^{A'C}\vp_{B)C}=& \;\Psi^{1/3}_2(\Box+2\Psi_2+\tfrac{2}{3}\l)[\Psi^{-1/3}_2\vp_{1}], \\
 2 \; \iota^{AB}\C_{A'(A}\C^{A'C}\vp_{B)C}=& \;\Psi^{2/3}_2(\teuk_{-2}-4\Psi_2+\tfrac{2}{3}\l)[\Psi^{-2/3}_2\vp_{2}].
\end{align}

\noindent
Spin $\s=2$:
\begin{align}
 \tfrac{d}{d\ve}|_{\ve=0}\{2\;o^{ABCD}\C_{A'(A}\C^{A'E}\vp_{BCD)E}\} = & \;(\teuk_{+4}-16\Psi_2+\tfrac{2}{3}\l)\dot\vp_0, \\
 \tfrac{d}{d\ve}|_{\ve=0}\{2\;o^{AB}\iota^{CD}\C_{A'(A}\C^{A'E}\vp_{BCD)E}\} = &
    \; \Psi^{2/3}_2(\Box+2\Psi_2+\tfrac{2}{3}\l)[\Psi^{-2/3}_2\dot\vp_2]  +\dot\Box_{\{-3;0,0\}}\vp_2, \label{linFI2} \\
 \tfrac{d}{d\ve}|_{\ve=0}\{2\;\iota^{ABCD}\C_{A'(A}\C^{A'E}\vp_{BCD)E}\} = & 
      \; \Psi^{4/3}_2(\teuk_{-4}-16\Psi_2+\tfrac{2}{3}\l)[\Psi^{-4/3}_2\dot\vp_4].
\end{align}

Finally, it is straightforward to show that
\begin{equation}\label{lhs-mi}
  \C_{A'_1(A_1}\C^{A'_1B}\vp_{A_2...A_n)B} = (\c_{A'_1(A_1}-nf_{A'_1(A_1})\c^{A'_1B}\vp_{A_2...A_n)B},
\end{equation}
thus we recover the main identity (\ref{mainid0}) (recall that by Remark \ref{remarkA} we have 
$f_a=A_a$ in type D).
Note that the formulation of section \ref{sec-cci} is actually more general since it deals with all 
algebraically special spacetimes, not just the type D.

\subsection{The conformally covariant Laplace-de Rham operator}\label{sec-laplace}

In the case of integer spin, the spinorial operator (\ref{mainspinorop}) admits a description in tensor terms. 
Since in curved spacetimes the only interesting cases of integer spin fields correspond to spin 1 and 2, we 
now briefly discuss the tensorial structure of the operator in these cases.

For 4-dimensional spacetimes, it was found in \cite{Araneda2017} that the tensor structure 
of the spinorial operator on the right hand side of (\ref{lhs-mi}) is that of a ``modified'' Laplace-de Rham 
operator acting on tensor valued differential forms. The same idea applies here, with the only difference 
that now the operator adopts a 
more `symmetric' form in terms of a covariant exterior derivative associated to $\C_{a}$. 
More precisely, for a tensor valued differential form (with well defined conformal weight)
$\w_{a_1...a_k b_1...b_l}=\w_{a_1...a_k [b_1...b_l]}$, we define 
the covariant exterior derivative $\mc{D}_{\C}$ and its adjoint $\mc{D}^{\dag}_{\C}$ by
\begin{align}
 & (\mc{D}_{\C}\w)_{a_1...a_k b_1...b_{l+1}}:=(l+1)\C_{[b_1}\w_{|a_1...a_k |b_2...b_{l+1}]}, \label{ced1} \\
 & (\mc{D}^{\dag}_{\C}\w)_{a_1...a_k b_1...b_{l-1}}:=-\C^{c}\w_{a_1...a_k c b_1...b_{l-1}}. \label{ced2}
\end{align}
Let $F_{ab}$ and $K_{abcd}$ be the tensorial analogues of the totally symmetric spinors $\vp_{AB}$ 
and $\vp_{ABCD}$, i.e. 
\begin{align}
 & F_{ab}=\vp_{AB}\bar\e_{A'B'}+\bar\vp_{A'B'}\e_{AB}, \\
 & K_{abcd}=\vp_{ABCD}\bar\e_{A'B'}\bar\e_{C'D'}+\bar\vp_{A'B'C'D'}\e_{AB}\e_{CD}
\end{align}
(where we are considering $K_{abcd}=K_{ab[cd]}$ as a tensor valued 2-form with the extra tensorial indices $ab$). 
$F_{ab}$ has conformal weight $w=0$, and $K_{abcd}$ has $w=1$.
Then:
\begin{lem}
With the definitions (\ref{ced1}) and (\ref{ced2}), we have the following identities:
\begin{align}
  -\tfrac{1}{2}[(\mc{D}_{\C}\mc{D}^{\dag}_{\C}+\mc{D}^{\dag}_{\C}\mc{D}_{\C})F]_{ab}
 = & \C_{E'(A}\C^{E'E}\vp_{B)E}\bar\e_{A'B'}+\C_{E(A'}\C^{EE'}\bar\vp_{B')E'}\e_{AB}, \\
\nonumber -\tfrac{1}{2}[(\mc{D}_{\C}\mc{D}^{\dag}_{\C}+\mc{D}^{\dag}_{\C}\mc{D}_{\C})K]_{abcd} = &
 \C_{E'(C}\C^{E'E}\vp_{|ABE|D)}\bar\e_{A'B'}\bar\e_{C'D'}\\
 & +\C_{E(C'}\C^{EE'}\bar\vp_{|A'B'E'|D')}\e_{AB}\e_{CD}.
\end{align}
\end{lem}

The 4-dimensional ``spin $\s$ modified'' Laplace-de Rham operator defined in \cite{Araneda2017} 
is then the conformally covariant Laplace-de Rham operator
\begin{equation}
 \mc{D}_{\C}\mc{D}^{\dag}_{\C}+\mc{D}^{\dag}_{\C}\mc{D}_{\C}
\end{equation}
acting on tensor valued differential forms with well defined conformal weight.

\section{Conclusions}

In this work we have developed a conformally and GHP covariant formalism for dealing with the massless 
free field equations (\ref{feq}) in (4-dimensional) algebraically special Einstein spacetimes, 
and we have shown that the operators associated to well known equations in the literature find a natural 
geometrical interpretation in this framework.
The main tool of the construction is the almost-complex structure (\ref{acs}) and its {\em unique} associated 
Weyl connection (\ref{uwc}) for the conformal manifold. 
Algebraically special spacetimes have preferred null directions on the geometry, and after adapting a null frame 
and its `gauge symmetry' to them we showed that the covariant derivative naturally induced (from the 
conformal structure) in associated vector bundles is precisely the Teukolsky connection, see 
Theorem \ref{thm-teuk}.
Furthermore, the almost-complex structure is integrable in type D spaces and then they become complex 
Hermitian manifolds.

A natural interpretation of the `ordinary' and `generalized' hidden symmetries (given respectively by 
(\ref{kseq}) and (\ref{mks})) persistent on black hole perturbations has also emerged from our formalism, 
since, on the one hand, by Lemma \ref{lem-kahler} the existence of a conformal Killing-Yano tensor 
(or its associated Killing spinor) in type D spacetimes can be thought to be a consequence of the 
presence of a K\"ahler metric in the conformal Hermitian class, and more generally, 
the Killing spinor equations are just the reflection of the type D principal spinors being parallel with 
respect to the natural Weyl-GHP connection, which in turn (by Lemma \ref{lem-ps}) is a consequence 
of algebraic speciality. 
In this sense the spin lowering/raising mechanism derivable from the identities given in section \ref{sec-cci} 
is closer in spirit to the original (and simpler) Penrose's work \cite{Penrose1965}, which uses covariantly 
constant spinors in Minkowski.

On the other hand, conformal spinor geometry is particularly well suited for the local twistor formalism, 
and we have shown that the preferred Weyl connection leads to a natural local twistor exact sequence (\ref{es2})
in the reversed order to the standard one, which results in the construction of {\em weighted} local twistors.
We introduced a connection on the weighted twistor bundle, and showed that the `weighted Killing spinor 
notion' deduced from it coincides with our earlier definition.

Finally, we mention that the generalized Teukolsky connection and the closely related `weighted Killing fields' 
found in \cite{Araneda2017} for perturbations of higher dimensional spacetimes, can be shown to
follow the same principle as the one exploited here, namely conformal and GHP covariance. 
However, the question about conformal invariance of field equations 
in higher dimensions is much more subtle than in the 4-dimensional case (in particular, Maxwell fields 
in $d\neq4$ are not conformally invariant).

\section*{Acknowledgements}

I would like to thank Gustavo Dotti, Martin Reiris and Oscar Reula for very helpful discussions. 
This work is supported by CONICET (Argentina).

\appendix

\section{Some details of calculations for spin 2}\label{appendix}

In this appendix we give some details of the calculations performed in section \ref{sec-spin2} 
corresponding to spin $\s=2$. We will use the notation and conventions of that section.

We are studying equations for the spin 2 field (\ref{rws}) in the 
context of gravitational perturbations. The field equations in this case are not $\c^{A'A}\vp_{ABCD}=0$, 
but the linearized Bianchi identities $\frac{d}{d\ve}|_{\ve=0}(\c^{A'A}\Psi_{ABCD})=0$. 
Replacing $\vp_{ABCD}=\Om^{-1}\Psi_{ABCD}$, we have:
\begin{align}
\nonumber \C_{A'A}\C^{A'E}\vp_{BCDE}=&\Om^{-1}(\c_{A'A}-4f_{A'A})\c^{A'E}\Psi_{BCDE}
 +\Psi_{BCDE}(\c_{A'A}-4f_{A'A})\c^{A'E}\Om^{-1} \\
& +(\c^{A'E}\Om^{-1})(\c_{A'A}\Psi_{BCDE})+(\c_{A'A}\Om^{-1})(\c^{A'E}\Psi_{BCDE}). \label{aux1}
\end{align}
Now let $\l^{ABCD}$ be some of the spinors in the set $\{o^{ABCD}, o^{(A}o^B\iota^C\iota^{D)}, \iota^{ABCD}\}$.
Projecting (\ref{aux1}) over $\l^{ABCD}$, linearizing, and taking into account the background 
assumptions $(\c^{A'A}\Psi_{ABCD})|_{\ve=0}=0$ and $(\c_{a}\Om)|_{\ve=0}=0$ 
(since $\Om(0)=\mr\Om=\rm{const.}$), we get 
\begin{align}
\nonumber \tfrac{d}{d\ve}|_{\ve=0} \{ \l^{ABCD}\C_{A'A}\C^{A'E}\vp_{BCDE} \}
  = & \mr\Om^{-1}\tfrac{d}{d\ve}|_{\ve=0} \{ \l^{ABCD}(\c_{A'A}-4f_{A'A})\c^{A'E}\Psi_{BCDE} \} \\
\nonumber & +\tfrac{d}{d\ve}|_{\ve=0} \{ \l^{ABCD}\Psi_{BCDE}(\c_{A'A}-4f_{A'A})\c^{A'E}\Om^{-1} \} \\
& +\tfrac{d}{d\ve}|_{\ve=0} \{ \l^{ABCD}(\c^{A'E}\Om^{-1})(\c_{A'A}\Psi_{BCDE}) \}. \label{aux2}
\end{align}

For $\l^{ABCD}=o^{ABCD}$, it is straightforward to see that, if the background spacetime is 
algebraically special (so that $\k(0)=\sigma(0)=\Psi_0(0)=\Psi_1(0)=0$)
the second and third lines in the right hand side of (\ref{aux2}) vanish, thus the left hand side 
of the identity (\ref{lce0}) vanishes on shell and we get the result mentioned in that section.

For $\l^{ABCD}=o^{(A}o^B\iota^C\iota^{D)}$, a tedious calculation assuming that the background spacetime 
is type D shows that (\ref{aux2}) becomes
\begin{align*}
\tfrac{d}{d\ve}|_{\ve=0} \{ o^{(A}o^B\iota^C\iota^{D)}\C_{A'A}\C^{A'E}\vp_{BCDE} \}
  = & \mr\Om^{-1}\tfrac{d}{d\ve}|_{\ve=0} \{ o^{(A}o^B\iota^C\iota^{D)}(\c_{A'A}-4f_{A'A})\c^{A'E}\Psi_{BCDE} \} \\
 & +\tfrac{1}{2}\Psi_2(\Box+2f^a\c_a)\dot\Om^{-1}.
\end{align*}
Now, using the definition of the conformal wave operator $\Box_{\{w;p,q\}}$ and some identities 
of the background (type D) spacetime, it is not difficult to show that
\begin{equation}
 \Psi_2(\Box+2f^a\c_a)\dot\Om^{-1}=\Box_{\{-3;0,0\}}[\Psi_2\dot\Om^{-1}],
\end{equation}
then the discussion between equations (\ref{lce2}) and (\ref{eqlwe}) follows.

\end{document}